\documentclass[11pt]{article}

\usepackage{hyperref}
\usepackage{amsmath} 
\usepackage{amsthm} 
\usepackage{amssymb}	
\usepackage{graphicx} 
\usepackage{multicol} 
\usepackage{multirow}
\usepackage{color}
\usepackage[dvips,letterpaper,margin=1in,bottom=1in]{geometry}
\usepackage[capitalize,noabbrev]{cleveref}

\usepackage[utf8]{inputenc}
\usepackage[english]{babel}

\usepackage{diagbox}
\usepackage{mathtools}

\newtheorem{theorem}{Theorem}
\newtheorem{lemma}{Lemma}
\newtheorem{corollary}{Corollary}

\newtheorem{definition}{Definition}

\newtheorem{remark}{Remark}
\newtheorem{problem}{Problem}

\DeclarePairedDelimiter\rbra{\lparen}{\rparen}
\DeclarePairedDelimiter\sbra{\lbrack}{\rbrack}
\DeclarePairedDelimiter\cbra{\{}{\}}
\DeclarePairedDelimiter\abs{\lvert}{\rvert}
\DeclarePairedDelimiter\Abs{\lVert}{\rVert}
\DeclarePairedDelimiter\ceil{\lceil}{\rceil}
\DeclarePairedDelimiter\floor{\lfloor}{\rfloor}
\DeclarePairedDelimiter\ket{\lvert}{\rangle}
\DeclarePairedDelimiter\bra{\langle}{\rvert}

\DeclareMathOperator*{\E}{\mathbb{E}}

\newcommand{\poly} {\operatorname{poly}}

\newcommand{\polylog} {\operatorname{polylog}}

\usepackage{latexsym}
\usepackage{CJK}

\usepackage{enumerate}

\usepackage{algorithm}
\usepackage{algpseudocode}

\usepackage{stmaryrd}
\usepackage{tabularx}
\usepackage{booktabs}

\newcommand{\footremember}[2]{%
    \footnote{#2}
    \newcounter{#1}
    \setcounter{#1}{\value{footnote}}%
}

\begin{document}
    \title{Tight Quantum Depth Lower Bound for\\ Solving Systems of Linear Equations}
\author{
    Qisheng Wang \footremember{1}{Qisheng Wang is with the Graduate School of Mathematics, Nagoya University, Nagoya, Japan (e-mail: \url{QishengWang1994@gmail.com}).}
    \and Zhicheng Zhang \footremember{2}{Zhicheng Zhang is with the Centre for Quantum Software and Information, University of Technology Sydney, Sydney, Australia (e-mail: \url{iszczhang@gmail.com}).}
}
        \date{}
        \maketitle

    \begin{abstract}
        Since \hyperlink{cite.HHL09}{Harrow, Hassidim, and Lloyd (2009)} showed that a system of linear equations with $N$ variables and condition number $\kappa$ can be solved on a quantum computer in $\poly\rbra{\log\rbra{N}, \kappa}$ time, exponentially faster than any classical algorithms, its improvements and applications have been extensively investigated.
        The state-of-the-art quantum algorithm for this problem is due to \hyperlink{cite.CAS+22}{Costa, An, Sanders, Su, Babbush, and Berry (2022)}, with optimal query complexity $\Theta\rbra{\kappa}$.
        An important question left is 
        whether parallelism can bring further optimization. 
        In this paper, we study the limitation of parallel quantum computing on this problem. 
        We show that any quantum algorithm for solving systems of linear equations with time complexity $\poly\rbra{\log\rbra{N}, \kappa}$ has 
        a lower bound of $\Omega\rbra{\kappa}$ on the depth of queries, which is tight up to a constant factor.
    \end{abstract}

    \textbf{Keywords: quantum computing, lower bounds, solving systems of linear equations, quantum query complexity, depth complexity, parallel computing.}

    \newpage
    \tableofcontents
    \newpage

    \section{Introduction}

    \paragraph*{Quantum linear systems problem.}
    Since the discovery of the celebrated quantum algorithm for solving systems of linear equations by \cite{HHL09}, it has been applied in various fields, e.g., machine learning \cite{BWP+17}, quantum chemistry \cite{CRO+19}, and finance \cite{OML19}. 

    The Quantum Linear Systems Problem (QLSP) is to prepare an $\varepsilon$-approximation to the quantum state $\ket{x}$ that is proportional to $A^{-1}\ket{b}$, given access to matrix $A$ and vector $\ket{b}$. 
    Since hard instances for QLSP are known by taking $\ket{b} = \ket{0}$ in \cite{HHL09}, we formally state (the special form of) QLSP as follows for simplicity. 
    \begin{problem} [QLSP]
        Suppose that $A \in \mathbb{C}^{N \times N}$ is an Hermitian matrix with known condition number $\kappa > 0$ such that $I/\kappa \leq A \leq I$. Let
        \[
        \ket{x} = \frac{A^{-1}\ket{0}}{\Abs{A^{-1}\ket{0}}}.
        \]
        Given quantum query access to $A$, the goal is to prepare a quantum state $\ket{\tilde x}$ such that $\Abs{\ket{\tilde x} - \ket{x}} \leq \varepsilon$ with probability at least $2/3$.
        We use $\textup{QLSP}\rbra{N, \kappa, \varepsilon}$ to denote the problem with the chosen parameters.
    \end{problem}
    Here, two types of quantum query access to a matrix are often considered in the literature: quantum query access to sparse matrices and to block-encoded matrices. 
    The former assumes a quantum oracle $\mathcal{O}_A$ that computes each entry of an $O\rbra{1}$-sparse matrix $A$ (given the  row and column indices) and a quantum oracle $\mathcal{O}_s$ that computes the index of each non-zero entry in each row (given the row index);
    the latter assumes a quantum oracle $U_A$ that is a block-encoding of a matrix $A$ (not necessarily sparse), i.e., roughly speaking, $A$ is encoded in the upper left corner of the unitary operator $U_A$. 
    Let $Q^{\textup{sparse}}\rbra{N, \kappa, \varepsilon}$ and $Q^{\textup{block}}\rbra{N, \kappa, \varepsilon}$ denote the quantum query complexity for $\textup{QLSP}\rbra{N, \kappa, \varepsilon}$ with access to sparse and block-encoded matrices, respectively. 
    
    Since quantum access to sparse matrices can be converted to quantum access to block-encoded matrices as noted in \cite{GSLW19}, it naturally holds that $Q^{\textup{sparse}}\rbra{N, \kappa, \varepsilon} = O\rbra{Q^{\textup{block}}\rbra{N, \kappa, \varepsilon}}$.
    For simplicity, here we only consider the quantum query complexity for QLSP in terms of sparse matrices.
    The first quantum algorithm for QLSP proposed in \cite{HHL09} is based on quantum phase estimation \cite{LP96,CEMM98,BDM99}, resulting in a query complexity of $\widetilde O\rbra{\kappa^2/\varepsilon}$ and a time complexity of $\poly\rbra{\log\rbra{N}, \kappa, 1/\varepsilon}$, where $\widetilde O\rbra{\cdot}$ suppresses logarithmic factors; they also gave an $\Omega\rbra{\kappa^{1-\delta} \polylog\rbra{N}}$ quantum time lower bound for QLSP.
    Shortly after, the query upper bound was improved to $\widetilde O\rbra{\kappa/\varepsilon^3}$ in \cite{Amb12} by variable-time amplitude amplification, with an almost optimal dependence on $\kappa$.
    An exponential improvement over the dependence on $\varepsilon$ was obtained in \cite{CKS17} via the Linear-Combinations-of-Unitaries (LCU) technique \cite{CW12,BCC+15} combined with quantum walks for Hamiltonians, achieving a query complexity of $O\rbra{\kappa \polylog\rbra{\kappa/\varepsilon}}$.
    Subsequent works then focused on optimizing the logarithmic factors in the complexity. 
    The query complexity for QLSP was improved to $O\rbra{\kappa \log \rbra{\kappa}/\varepsilon}$ in \cite{SSO19} based on the adiabatic randomization method; later, it was improved to $O\rbra{\kappa \log^2\rbra{\kappa} \log^4\rbra{{\log \rbra{\kappa}}/{\varepsilon}}}$ in \cite{AL22} based on the time-optimal adiabatic method, and to $O\rbra{ \kappa \rbra{ {\log\rbra{\kappa}}/ {\log\rbra{\log\rbra{\kappa}}} + \log\rbra{1/ \varepsilon} } }$ in \cite{LT20} based on Zeno eigenstate filtering.
    In \cite{OD21}, it was shown that $\Omega\rbra{\kappa}$ queries are required to solve QLSP even if the matrix is positive definite; they also identified a class of positive definite matrices for which efficient quantum algorithms with query complexity $\widetilde O\rbra{\sqrt{\kappa}}$ exist.
    Recently, logarithmic factors of $\kappa$ was finally removed in \cite{CAS+22}, resulting in a query complexity of $O\rbra{\kappa \log\rbra{1/\varepsilon}}$, which is optimal according to the lower bound $\Omega\rbra{\kappa \log\rbra{1/\varepsilon}}$ claimed in the forthcoming work \cite{HK21}.

    \begin{theorem} [Optimal QLSP solver, {\cite{CAS+22,HK21}}]
    \label{thm:qlsp}
        $Q^{\textup{sparse}}\rbra{N, \kappa, \varepsilon} = \Theta\rbra{\kappa \log\rbra{1/\varepsilon}}$. 
    \end{theorem}

    We summarize the developments of QLSP in \cref{tab:qlsp}.

    \begin{table}[!htp]
    \centering
    \caption{Developments on the quantum query complexity for QLSP.}
    \label{tab:qlsp}
    \begin{tabular}{ccc}
    \toprule
    Year & $Q^{\textup{sparse}}\rbra{N, \kappa, \varepsilon}$ & Reference \\ \midrule
    2008 & $\widetilde O\rbra{\kappa^2/\varepsilon}$, time $\geq \Omega\rbra{\kappa^{1-\delta}\polylog\rbra{N}}$ & \cite{HHL09}     \\ 
    2010 & $\widetilde O\rbra{\kappa/\varepsilon^3}$ & \cite{Amb12}     \\ 
    2015 & $O\rbra{\kappa \polylog\rbra{\kappa/\varepsilon}}$ & \cite{CKS17}     \\ 
    2018 & $O\rbra{\kappa \log \rbra{\kappa}/\varepsilon}$ & \cite{SSO19}     \\ 
    2019 & $O\rbra{\kappa \log^2\rbra{\kappa} \log^4\rbra{{\log \rbra{\kappa}}/{\varepsilon}}}$ & \cite{AL22}      \\ 
    2019 & $O\rbra{ \kappa \rbra{ {\log\rbra{\kappa}}/ {\log\rbra{\log\rbra{\kappa}}} + \log\rbra{1/ \varepsilon} } }$ & \cite{LT20}      \\ 
    2021 & $\Omega\rbra{\kappa}$ & \cite{OD21} \\
    2021 & $O\rbra{\kappa \log\rbra{1/\varepsilon}}$ & \cite{CAS+22}    \\ 
    2021 & $\Omega\rbra{\kappa \log\rbra{1/\varepsilon}}$ & \cite{HK21} \\
    2024 & $\geq Q_{\parallel}^{\textup{sparse}}\rbra{N, \kappa, \varepsilon} \geq 0.249 \kappa$ & This work \\
    \bottomrule
    \end{tabular}
    \end{table}

    \paragraph*{Parallel quantum computation.}

    The power of parallelism in quantum computing first attracted researchers' attention due to a fast parallel circuit of the quantum Fourier transform \cite{CW00}, which can be used to implement Shor's factoring algorithm \cite{Sho94} with quantum logarithmic depth and classical polynomial-time pre- and post-processing. 
    This breakthrough was later implemented in 2D nearest-neighbor quantum architecture \cite{PS13}, and was extended to finding discrete logarithms on ordinary binary elliptic curves \cite{RS14}.
    The parallel version of Grover search \cite{Gro96} and its extensions were also studied in \cite{Zal99,GWC00,GR04,Bur19,GTKL+22}.
    Inspired by this, an optimal parallel quantum query algorithm for element distinctness \cite{Amb07} was proposed in \cite{JMdW17}.
    Recently, an unconditional quantum advantage with constant-depth circuits was discovered in \cite{BGK18}.
    This surprising result was further enhanced in \cite{BWKST19,LG19,BGKT20,GS20,CSV21,BWP23}.
    A constant-depth quantum circuit for multivariate trace estimation was given in \cite{QKW22}.
    A low-depth Hamiltonian simulation was proposed in \cite{ZWY21} for a class of Hamiltonians; also, a depth lower bound $\Omega\rbra{t}$ was recently shown in \cite{CCH+23} for Hamiltonian simulation for time $t$.
    Optimal space-depth trade-offs were found for CNOT circuits \cite{JST+20} and quantum state preparation \cite{STY+23,Ros21,ZLY22,YZ23}.

    In complexity theory, the complexity classes $\mathsf{QNC}$ and $\mathsf{QACC}$, the quantum analogs of classical classes $\mathsf{NC}$ and $\mathsf{ACC}$ (for problems efficiently computable in parallel), were first defined in \cite{MN01} and \cite{GHP00}, respectively, and they were further studied in the literature \cite{GHMP02,TD04,FGHZ05,HS05,TT16}.
    Inspired by the discovery of classical-quantum hybrid algorithm for factoring \cite{CW00}, Jozsa \cite{Joz06} and Aaronson \cite{Aar05} raised the open problem of whether any polynomial-time quantum algorithm can be simulated by a polynomial-time classical algorithm interleaved with low-depth quantum computation. 
    This problem aims to compare the computability between $\mathsf{BQP}$, $\mathsf{BPP}^{\mathsf{BQNC}}$, and $\mathsf{BQNC}^{\mathsf{BPP}}$.
    Recently, this problem was answered by \cite{CCL23,CM20}, giving an oracle separation of $\mathsf{BQP}$ from either $\mathsf{BPP}^{\mathsf{BQNC}}$ or $\mathsf{BQNC}^{\mathsf{BPP}}$, which was further improved by \cite{AGS22,HLG22,ACC+23}. 

    Parallelism was also investigated in measurement-based quantum computing \cite{BK09,BKP10}, distributed quantum computing \cite{YF09,BBG+13}, and quantum programming \cite{YZLF22}. 

    \subsection{Main results}

    As mentioned above, a number of quantum algorithms turned out to have low-depth circuit implementations, which have potentials to be realized in the NISQ era \cite{Pre18}. 
    One may wonder if systems of linear equations can be solved in parallel on a quantum computer. 
    We study the limitation of quantum parallelism on this problem. 
    We use $Q_{\parallel}^{\textup{sparse}}\rbra{N, \kappa, \varepsilon}$ and $Q_{\parallel}^{\textup{block}}\rbra{N, \kappa, \varepsilon}$ to denote the quantum query-depth complexity for QLSP with quantum access to sparse matrices and to block-encoded matrices, respectively. 
    Here, the quantum query-depth complexity for QLSP means the minimal depth (with respect to queries) of the quantum circuit that solves QLSP in quantum time $\poly\rbra{\log\rbra{N}, \kappa, 1/\varepsilon}$ (see \cref{sec:def-Q-QLSP} for the formal definition). 

    Our main result is as follows.
    \begin{theorem} [Tight depth for sparse-QLSP, \cref{thm:main}]
    \label{thm:intro}
        For every constant $0 < \varepsilon < 0.015$, we have 
        $Q_{\parallel}^{\textup{sparse}}\rbra{N, \kappa, \varepsilon} \geq 0.249 \kappa = \Omega\rbra{\kappa}$.~That is, any quantum algorithm for sparse-QLSP with time complexity $\poly\rbra{\log\rbra{N}, \kappa, 1/\varepsilon}$ has depth (with respect to queries) at least $0.249\kappa$.
    \end{theorem}

    We compare our result with the known bounds in \cref{tab:qlsp}.
    As a corollary, we can also obtain a tight depth lower bound for QLSP with quantum access to block-encoded matrices. 

    \begin{corollary} [Tight depth for block-QLSP, \cref{thm:block-QLSP}]
    \label{corollary:intro-block}
        For every constant $0 < \varepsilon < 0.015$, we have 
        $Q_{\parallel}^{\textup{block}}\rbra{N, \kappa, \varepsilon} \geq 0.031 \kappa = \Omega\rbra{\kappa}$.~That is, any quantum algorithm for block-QLSP with time complexity $\poly\rbra{\log\rbra{N}, \kappa, 1/\varepsilon}$ has depth (with respect to queries) at least $0.031\kappa$. 
    \end{corollary}
    
    \cref{thm:intro} and \cref{corollary:intro-block} give tight lower bounds $\Omega\rbra{\kappa}$ for QLSP only up to a constant factor because, e.g., $Q_{\parallel}^{\textup{sparse}}\rbra{N, \kappa, \varepsilon} \leq Q^{\textup{sparse}}\rbra{N, \kappa, \varepsilon} = O\rbra{\kappa}$ for constant $\varepsilon > 0$ by \cref{thm:qlsp}.
    This implies that there is no low-depth quantum algorithm for QLSP in general. 

    Regarding the range of the constant $\varepsilon$ in \cref{thm:intro} and \cref{corollary:intro-block}, it is worth noting that if one can solve QLSP to a constant precision $\varepsilon_0 = \Theta\rbra{1}$, then one can improve this to arbitrarily small precision $\varepsilon$ at an additional cost of $O\rbra{\kappa\log\rbra{1/\varepsilon}}$ via ``eigenstate filtering'' \cite{LT20}.

    \paragraph*{Constant factor in quantum complexity.}

    It was shown in \cite{HHL09} that any quantum query algorithm for QLSP has time complexity $\Omega\rbra{\kappa^{1-\delta}\polylog\rbra{N}}$ for any $\delta > 0$.
    Then, it was shown in \cite{OD21} that $\Omega\rbra{\kappa}$ queries are necessary for QLSP even if the matrix is positive definite. 
    Recently, a matching quantum query lower bound $\Omega\rbra{\kappa \log\rbra{1/\varepsilon}}$ for QLSP was claimed in the forthcoming work \cite{HK21}.
    These results did not consider the explicit constant factor hidden in their lower bounds. 
    As a corollary, our quantum depth lower bounds for QLSP imply quantum query lower bounds with an explicit multiplicative constant factor as follows. 
    \begin{corollary}
    \label{corollary:intro-query}
        $Q^{\textup{sparse}}\rbra{N, \kappa, \varepsilon} \geq 0.249 \kappa$ and $Q^{\textup{block}}\rbra{N, \kappa, \varepsilon} \geq 0.031 \kappa$.
    \end{corollary}
    \begin{proof}
        This is straightforward by noting that any depth lower bound is also a query lower bound, e.g., $Q^{\textup{sparse}}\rbra{N, \kappa, \varepsilon} \geq Q^{\textup{sparse}}_{\parallel}\rbra{N, \kappa, \varepsilon}$.
        Taking the bounds in \cref{thm:intro} and \cref{corollary:intro-block} leads to the conclusion.
    \end{proof}
    For comparison, a quantum query upper bound $Q^{\textup{block}}\rbra{N, \kappa, \varepsilon} \leq \rbra*{117235 + \log\rbra{2/\varepsilon}} \kappa + O\rbra*{\sqrt{\kappa}}$ for block-QLSP was recently derived in \cite{CAS+22,JLP+23}.
    Combining it with the lower bound given in \cref{corollary:intro-query}, we conclude that the quantum query complexity $Q^{\textup{block}}\rbra{N, \kappa, \varepsilon}$ grows linearly in $\kappa$ with a constant factor bounded by
    \begin{equation}
    \label{eq:constants}
    0.031 \leq \frac{Q^{\textup{block}}\rbra{N, \kappa, \varepsilon}}{\kappa} \leq 117235 + \log\rbra{2/\varepsilon},
    \end{equation}
    as $\kappa \to \infty$.
    Here, we bound the ratio $Q^{\textup{block}}\rbra{N, \kappa, \varepsilon}/\kappa$ from both sides with explicit constants, 
    though the range still remains large. 
    The right hand side the \cref{eq:constants} can be improved to $56.0\kappa + 1.05 \log\rbra{1/\varepsilon}$ due to the very recent work \cite{Dal24}.

    \subsection{Techniques}

    The idea for proving the lower bounds is to reduce some computational problem that is hard to solve in parallel on a quantum computer to QLSP.
    That is, the reduction means that any low-depth quantum algorithm for QLSP can be used to solve this hard problem in parallel.
    Previously, the quantum time lower bound $\Omega\rbra{\kappa^{1-\delta}\polylog\rbra{N}}$ for QLSP was derived in \cite{HHL09} by reducing the simulation of $\mathsf{BQP}$ quantum circuits to QLSP. 
    In our case, we choose to reduce the \textit{permutation chain} problem, which was discovered in \cite{CCH+23} and recently shown to be hard for parallel quantum computing. 
    The permutation chain problem is to find the final element after a chain of permutations applied on the initial element $0$; that is,
    to find $\Pi_q(0)=\rbra{\pi_q \circ \cdots \circ \pi_2 \circ \pi_1}(0)$, given quantum query access to the ($0$-indexed) permutation $\pi_j$ of size $N = 2^n$ for every $1 \leq j \leq q = \poly\rbra{n}$ (see \cref{sec:permchain} for the formal definition).
    It was shown in \cite{CCH+23} that any quantum algorithm for permutation chain with time complexity $\poly\rbra{n}$ has depth $\Omega\rbra{q}$.

    To obtain a depth lower bound for QLSP, we need to encode the permutation chain into a system of linear equations (see \cref{sec:encode-permchain-by-QLSP} for the explicit construction). 
    Our construction is based on that of \cite{HHL09}. 
    In comparison, the condition number $\kappa$ of the constructed matrix is bounded by the size of the simulated quantum circuit in \cite{HHL09}; while, in our case, $\kappa$ is bounded by the length $q$ of the permutation chain.
    The difference is that since we only consider the query complexity but not gate complexity here, our construction of the system of linear equations is stand-alone and does not directly relate to quantum computing. 

    Let $A$ be the Hermitian matrix that encodes the permutation chain as above. It can be shown that measuring the normalized solution $\ket{x} \propto A^{-1}\ket{0}$ can produce the desired result $\Pi_q\rbra{0}$ with a constant probability of, say, at least $0.01$. 
    Using the standard success probability amplification, we can obtain $\Pi_q\rbra{0}$ with probability at least $2/3$ by $O\rbra{1}$ repetitions of a QLSP solver. 
    Therefore, any QLSP solver should have depth $\Omega\rbra{q}$ due to the hardness of the permutation chain, where $q = \Theta\rbra{\kappa}$ interchangeably. 

    We note that the permutation chain problem was also used in \cite{CCH+23} to derive the quantum depth lower bound for Hamiltonian simulation, but our idea and techniques are different from theirs.
    For comparison, they used a graph-to-Hamiltonian reduction which solves the permutation chain using a quantum walk on a line \cite{CCD+03}; in our case, we construct a QLSP encoding of a permutation chain, which modifies the construction of \cite{HHL09}.

    \paragraph*{Towards explicit constant factors.} 
    To derive an explicit (multiplicative) constant factor in the complexity, we first observe that the constructed matrix $A$ has condition number $\kappa \leq 2.001q$ for sufficiently large $\kappa$. 
    Second, each of the quantum oracles $\mathcal{O}_s$ and $\mathcal{O}_{A}$ for sparse matrix $A$ can be implemented using only $1$ query to the quantum oracle $\mathcal{O}_{\pi}$ for the permutations $\pi_1, \pi_2, \dots, \pi_q$. 
    Finally, combining the lower bound $q/2$ for the permutation chain \cite{CCH+23} (see \cref{thm:permchain}) with the above construction, we conclude that any QLSP should have depth $\geq \kappa/4.002 \geq 0.249 \kappa$ as stated in \cref{thm:intro} for sparse-QLSP. 

    For block-QLSP, the depth lower bound is obtained through the construction of block-encoding from the sparse input model given in \cite{GSLW19}.
    Specifically, our constructed matrix is $2$-sparse, and thus its (scaled) unitary block-encoding can be implemented using $4$ queries to the sparse oracles $\mathcal{O}_s$ and $\mathcal{O}_A$ in total. 
    From this, we can establish a connection from the complexity of block-QLSP to that of sparse-QLSP with explicit multiplicative constant factors. 
    The constant factor in \cref{corollary:intro-block} is then obtained after some detailed error analysis. 

    \subsection{Discussion}

    We study the limitation of parallel quantum computing for solving systems of linear equations and give a matching quantum depth lower bound $\Omega\rbra{\kappa}$.
    This means that quantum algorithms for QLSP cannot be parallelized in general if one hopes to retain the quantum exponential speedup in the dimension $N$ of the matrix.
    We conclude by mentioning some open questions regarding the quantum complexity for solving systems of linear equations. 
    \begin{enumerate}
        \item Since the constant factor of $\kappa$ still remains in a large range as shown in \cref{eq:constants}, an immediate question is whether we can tighten the range of the constant factor.
        The constant factor is important in the near future as QLSP is one of the most promising applications of quantum computing, with detailed running costs very recently analyzed in \cite{JLP+23}. 
        \item In this paper, we only obtain a depth lower bound in terms of the condition number $\kappa$. Can we derive a (joint) depth lower bound in terms of the precision $\varepsilon$ (and condition number $\kappa$)?
        For reference, a tight quantum query lower bound $\Omega\rbra{\kappa\log\rbra{1/\varepsilon}}$ for QLSP has been claimed in the forthcoming work \cite{HK21}.
        \item In \cref{thm:intro}, the depth lower bound holds when the dimension of the matrix $A$ is exponential in $\kappa$.
        This requirement is due to the reduction from the permutation chain problem (\cref{def:pc}) to QLSP in \cref{sec:reduction}, where the negligibility of the probability in \cref{eq:expected-succ-prob} requires a symmetric group of exponentially large degree.
        Given this drawback, the lower bound in \cref{thm:intro} holds only for $\kappa = O\rbra{\polylog N}$ (see \cref{rmk:kappa}).
        On the other hand, the quantum query lower bound $\Omega\rbra{\kappa}$ for QLSP given in \cite{OD21} holds for all $\kappa \leq N$.
        An interesting question is whether we can broaden the range of $\kappa$ in the quantum depth lower bound for QLSP.
    \end{enumerate}

    \subsection{Organization}

    We will introduce the necessary preliminaries in \cref{sec:preliminaries}.
    The quantum depth lower bounds for sparse-QLSP and block-QLSP will be derived in \cref{sec:sparse-QLSP} and \cref{sec:block-QLSP}, respectively.

    \section{Preliminaries}
    \label{sec:preliminaries}

    We first introduce basic notations, then define the quantum complexity for QLSP, and, finally, include some useful lemmas. 

    \subsection{Block-encoding}

    Block-encoding is a useful concept to describe quantum operators encoded as blocks in other quantum operators. 

    \begin{definition} [Block-encoding]
        \label{def:block-encoding}
        Suppose that $A$ is an $n$-qubit operator, $\alpha, \varepsilon \geq 0$ and $a \in \mathbb{N}$. 
        An $\rbra{n+a}$-qubit operator $B$ is an $\rbra{\alpha, a, \varepsilon}$-block-encoding of $A$, if 
        \[
        \Abs*{\alpha \bra{0}^{\otimes a} B \ket{0}^{\otimes a} - A} \leq \varepsilon,
        \]
        where $\Abs{\cdot}$ is the operator norm. 
    \end{definition}

    \subsection{Quantum query access}

    Let $A \in \mathbb{C}^{N \times N}$ be an Hermitian matrix with $\Abs{A} \leq 1$. 
    There are two main types of input models for quantum query access to $A$.

    \paragraph*{Sparse input model.}

    Suppose $A$ is an $s$-sparse Hermitian matrix, i.e., there are at most $s$ non-zero entries in each row and column of $A$.
    The sparse input model consists of two quantum oracles $\mathcal{O}_s$ and $\mathcal{O}_A$.
    Here, $\mathcal{O}_s$ computes the column index $l_{j, k}$ of the $k$-th non-zero entry in the $j$-th row of $A$ for $1 \leq j \leq N$ and $1 \leq k \leq s$, i.e., 
    \begin{equation} \label{eq:def-Os}
    \mathcal{O}_s \ket{j, k} = \ket{j, l_{j,k}},
    \end{equation}
    and $\mathcal{O}_A$ computes the $k$-th entry in the $j$-th row of $A$ for $1 \leq j \leq N$ and $1 \leq k \leq N$, i.e., 
    \begin{equation} \label{eq:def-OA}
    \mathcal{O}_A \ket{j, k, 0} = \ket{j, k, A_{j, k}}.
    \end{equation}
    Here, we assume that $A_{j, k}$ is in a binary representation and, for simplicity, we assume that the binary representation is exact.

    \paragraph*{Block-encoded input model.}

    The block-encoded input model consists of a quantum oracle $U_A$ and its inverse and controlled versions.
    Here, $U_A$ is a $\rbra{1, a, 0}$-block-encoding of $A$ for some $a = \polylog\rbra{N}$. 
    In the following, we give a construction of block-encoding from the sparse input model.

    \begin{lemma} [Block-encoding of sparse matrices, {\cite[Lemma 48 in the full version]{GSLW19}}]
    \label{lemma:sparse-to-block}
        Given quantum query access to $s$-sparse matrix $A \in \mathbb{C}^{N \times N}$ in the sparse input model by $\mathcal{O}_s$ and $\mathcal{O}_A$ with $\abs{A_{j,k}} \leq 1$, we can implement a quantum circuit that is an $\rbra{s, \ceil{\log_2\rbra{N}}+3, \varepsilon}$-block-encoding of $A$, using $2$ queries to $\mathcal{O}_s$, $2$ queries to $\mathcal{O}_A$, and $O\rbra{\log\rbra{N}+\log^{2.5}\rbra{s/\varepsilon}}$ one- and two-qubit quantum gates.
    \end{lemma}

    \subsection{Quantum query complexity for QLSP}
    \label{sec:def-Q-QLSP}

    A quantum query algorithm $\mathcal{A}$ given access to certain quantum oracles is described by a quantum circuit
    \[
    \mathcal{A} = G_Q U_Q G_{Q-1} \dots G_2 U_2 G_1 U_1 G_0,
    \]
    where each $U_j$ is a (controlled-)oracle and each $G_j$ is a quantum gate independent of the input oracles. 
    The quantum query complexity of $\mathcal{A}$ is defined to be $\mathsf{Q}\rbra{\mathcal{A}} = Q$, and the quantum time complexity of $\mathcal{A}$ is defined to be $\mathsf{T}\rbra{\mathcal{A}} = Q+\sum_{j=0}^Q \mathsf{C}\rbra{G_j}$, where $\mathsf{C}\rbra{G_j}$ denotes the number of one- and two-qubit quantum gates to implement $G_j$.
    
    Now that quantum query algorithms for QLSP with time complexity $\poly\rbra{\log\rbra{N}, \kappa, 1/\varepsilon}$, and even $\poly\rbra{\log\rbra{N}, \kappa, \log\rbra{1/\varepsilon}}$, are known, e.g., \cite{HHL09,CKS17}, achieving an exponential speedup in $N$ over classical algorithms, we are only concerned with those quantum algorithms with such exponential speedup.
    We consider two variants of QLSP as follows.
    \begin{problem} [Sparse-QLSP]
        Let $A \in \mathbb{C}^{N \times N}$ be an $O\rbra{1}$-sparse Hermitian matrix with $I/\kappa \leq A \leq I$. 
        Given quantum oracles $\mathcal{O}_s$ and $\mathcal{O}_A$ for quantum query access to $A$, the task is to solve $\textup{QLSP}\rbra{N, \kappa, \varepsilon}$ for $A$. 
    \end{problem}
    \begin{problem} [Block-QLSP]
        Let $A \in \mathbb{C}^{N \times N}$ be an Hermitian matrix with $I/\kappa \leq A \leq I$. 
        Given quantum oracle $U_A$ that is a block-encoding of $A$, the task is to solve $\textup{QLSP}\rbra{N, \kappa, \varepsilon}$ for $A$. 
    \end{problem}
    For this purpose, the quantum query complexity for QLSP is defined by
    \begin{align*}
        Q^{\textup{sparse}}\rbra{N, \kappa, \varepsilon} & = \min \{\mathsf{Q}\rbra{\mathcal{A}}:  
        \mathcal{A} \text{ solves $\text{sparse-QLSP}\rbra{N, \kappa, \varepsilon}$ and } 
        \text{$\mathsf{T}\rbra{\mathcal{A}} = \poly\rbra{\log\rbra{N}, \kappa, 1/\varepsilon}$} \}, \\
        Q^{\textup{block}}\rbra{N, \kappa, \varepsilon} & = \min \{\mathsf{Q}\rbra{\mathcal{A}} : 
        \mathcal{A} \text{ solves $\text{block-QLSP}\rbra{N, \kappa, \varepsilon}$ and }
        \mathsf{T}\rbra{\mathcal{A}} = \poly\rbra{\log\rbra{N}, \kappa, 1/\varepsilon} \}.
    \end{align*}
    
    If a quantum query algorithm $\mathcal{A}$ given quantum oracle $\mathcal{O}$ can be described as
    \[
    \mathcal{A} = F_{D} V_{D} F_{D-1} \dots F_2 V_2 F_1 V_1 F_0,
    \]
    where $V_j = \mathcal{O}^{\otimes k} \otimes I$ for some $k$ (we also call such $V_j$ a $k$-parallel query) or its controlled-version, and $F_j$ is a quantum gate independent of the oracle $\mathcal{O}$, then the quantum query-depth complexity (namely, the quantum depth complexity with respect to queries) of $\mathcal{A}$ is defined to be $\mathsf{Q}_{\parallel}\rbra{\mathcal{A}} = D$,\footnote{When $k$ is fixed, the query-depth complexity $\mathsf{Q}_{\parallel}\rbra{\cdot}$ is also known as the $k$-query complexity, denoted by $\mathsf{Q}^{k\parallel}\rbra{\cdot}$ in \cite{JMdW17}.} and the quantum time complexity is defined to be $\mathsf{T}\rbra{\mathcal{A}} = kD + \sum_{j=0}^D \mathsf{C}\rbra{F_j}$.
    The quantum query-depth complexity for QLSP is defined by
    \begin{align*}
        Q_{\parallel}^{\textup{sparse}}\rbra{N, \kappa, \varepsilon} & = \min \{\mathsf{Q}_{\parallel}\rbra{\mathcal{A}}:  
        \mathcal{A} \text{ solves $\text{sparse-QLSP}\rbra{N, \kappa, \varepsilon}$ and } 
        \text{$\mathsf{T}\rbra{\mathcal{A}} = \poly\rbra{\log\rbra{N}, \kappa, 1/\varepsilon}$} \}, \\
        Q_{\parallel}^{\textup{block}}\rbra{N, \kappa, \varepsilon} & = \min \{\mathsf{Q}_{\parallel}\rbra{\mathcal{A}} : 
        \mathcal{A} \text{ solves $\text{block-QLSP}\rbra{N, \kappa, \varepsilon}$ and} 
        \mathsf{T}\rbra{\mathcal{A}} = \poly\rbra{\log\rbra{N}, \kappa, 1/\varepsilon} \}.
    \end{align*}
    
    \subsection{Permutation chain}
    \label{sec:permchain}

    A formal definition of the permutation chain problem is given as follows.

    \begin{problem} [Permutation chain] \label{def:pc}
        Suppose that $N$ is a positive integer and $S_N$ is the symmetric group $S_N$ of degree $N$. 
        Let $q$ be a positive integer such that $q = \polylog\rbra{N}$. 
        Let $\pi_1, \pi_2, \dots, \pi_q \in S_N$ be ($0$-indexed) permutations of size $N$. 
        Suppose we are given the quantum oracle $\mathcal{O}_{\pi}$ such that for every $1 \leq j \leq 2q$ and $0 \leq x < N$, 
        \begin{equation} \label{eq:def-Opi}
        \mathcal{O}_{\pi} \ket{j, x, 0} = \begin{cases}
            \ket{j, x, \pi_j\rbra{x}}, & 1 \leq j \leq q, \\
            \ket{j, x, \pi_{j-q}^{-1}\rbra{x}}, & q + 1 \leq j \leq 2q.
        \end{cases}
        \end{equation}
        The goal is to compute $\Pi_q \rbra{0}$, where $\Pi_q = \pi_q \circ \dots \circ \pi_2 \circ \pi_1$.
        We use $\textup{PermChain}\rbra{N, q}$ to denote the problem with the chosen parameters. 
    \end{problem}

    Recently, the quantum hardness of the permutation chain in terms of circuit depth was shown in \cite{CCH+23} for the average case. 

    \begin{theorem} [Permutation chain, {\cite[Theorem 5.2]{CCH+23}}]
    \label{thm:permchain}
        Suppose that $N$ is a positive integer. 
        Let $q = O\rbra{\polylog\rbra{N}}$ and $k = O\rbra{\polylog\rbra{N}}$ be positive integers.
        For any quantum query algorithm $\mathcal{A}$ for $\textup{PermChain}\rbra{N, q}$ using $\floor{\rbra{q-1}/{2}}$ $k$-parallel queries to $\mathcal{O}_{\pi}$, we have
        \[
        \E_{\pi_1, \pi_2, \dots, \pi_q} \sbra*{ \Pr\sbra*{\mathcal{A} \text{ outputs } \Pi_q\rbra{0}} } = O\rbra*{q\sqrt{\frac k N}}.
        \]
        That is, in the average case, $\mathcal{A}$ finds the answer $\Pi_q\rbra{0}$ with probability $O\rbra{q\sqrt{k/N}}$.
    \end{theorem}

    \section{Quantum Depth Lower Bound for Sparse-QLSP}
    \label{sec:sparse-QLSP}

    In this section, we will derive the quantum depth lower bound for sparse-QLSP stated as follows.

    \begin{theorem}
    \label{thm:main}
        For every constant $0 < \varepsilon < \frac{e^{-4}-e^{-6}}{1-e^{-6}} \approx 0.015$, we have
        \[
        \limsup_{\kappa \to +\infty} \frac{Q_{\parallel}^{\textup{sparse}}\rbra{3\kappa 2^{\frac{\kappa}{2}}, \kappa, \varepsilon}}{\kappa} \geq \frac 1 4.
        \]
    \end{theorem}

    In other words, \cref{thm:main} means that for every $\delta > 0$, there is a large enough real number $\kappa$ such that 
    \[
        Q_{\parallel}^{\textup{sparse}}\rbra{3\kappa 2^{\frac{\kappa}{2}}, \kappa, \varepsilon} \geq \rbra*{\frac 1 4 - \delta} \kappa.
    \]
    Taking $\delta = 0.001$ will produce the introductory \cref{thm:intro}.
    We will prove \cref{thm:main} in the following subsections. 
    Here, we sketch the main idea of the proof.
    \begin{enumerate}
        \item We first encode the problem $\textup{PermChain}\rbra{N, q}$ as a system of linear equations $A\ket{x} = \ket{0}$ in \cref{sec:encode-permchain-by-QLSP}. 
        \item Then, we construct the quantum oracles $\mathcal{O}_s$ and $\mathcal{O}_A$ (in the sparse input model) for $A$ using the quantum oracle $\mathcal{O}_{\pi}$ for the permutation chain in \cref{sec:construct-access}.
        \item Finally, we reduce $\textup{PermChain}\rbra{N, q}$ to \textup{QLSP} with certain parameters in \cref{sec:reduction}.
    \end{enumerate}

    \subsection{Encoding permutation chain by linear systems}
    \label{sec:encode-permchain-by-QLSP}

    Let $\pi_1, \pi_2, \dots, \pi_q \in S_N$ be ($0$-indexed) permutations of size $N$. 
    Then, we consider the system of linear equations described by an Hermitian matrix $A$ such that
    \begin{equation}
    \label{eq:def-A}
    \begin{aligned}
    A = \frac{1}{1+e^{-\frac 1 q}} \Bigg( \ket{0}\bra{1} \otimes \rbra*{I - e^{-\frac{1}{q}} P} 
    + \ket{1}\bra{0} \otimes \rbra*{I - e^{-\frac{1}{q}} P^{-1}} \Bigg),
    \end{aligned}
    \end{equation}
    where $P$ is a $3qN \times 3qN$ permutation matrix defined by
    \begin{align*}
        P & = \sum_{j=1}^q \ket{j} \bra{j-1} \otimes \sum_{x=0}^{N-1} \ket{\pi_j\rbra{x}} \bra{x} 
        + \sum_{j=q+1}^{2q} \ket{j} \bra{j-1} \otimes I \\ 
        & \qquad + \sum_{j=2q+1}^{3q} \ket{j \bmod 3q} \bra{j-1} \otimes \sum_{x=0}^{N-1} \ket{\pi_{3q+1-j}^{-1}\rbra{x}} \bra{x}.
    \end{align*}
    The construction of $A$ in \cref{eq:def-A} is inspired by \cite{HHL09} where they used linear systems to describe the simulation of quantum computation in the quantum circuit model.
    For comparison, our construction encodes a precise problem --- $\textup{PermChain}\rbra{N, q}$, which is not directly related to quantum computing.
    It can be shown that $A$ satisfies the following basic properties. 
    \begin{lemma}
    \label{lemma:basic-A}
        Let $A$ be defined by \cref{eq:def-A}.
        Then, $A$ is a $2$-sparse $6qN \times 6qN$ Hermitian matrix.
        Moreover, for every $\delta > 0$, $\rbra{2+\delta}^{-1}q^{-1} I \leq A \leq I$ for sufficiently large $q \geq 1$. 
    \end{lemma}
    \begin{proof}
    It is easy to see that $A$ is a $2$-sparse $6qN \times 6qN$ Hermitian matrix. 
    Then, we consider how to bound $\Abs{A}$ and $\Abs{A^{-1}}$. 
    As $\Abs{A - B} \leq \Abs{A} + \Abs{B}$, we have
    \begin{align*}
        \Abs{A}
        & \leq \frac{1}{1+e^{-\frac 1 q}} \max\cbra*{\Abs*{I - e^{-\frac{1}{q}} P}, \Abs*{I - e^{-\frac{1}{q}} P^{-1}}} \\
        & \leq \frac{1}{1+e^{-\frac 1 q}} \rbra*{\Abs{I} + e^{-\frac{1}{q}} \Abs{P}} \\
        & = \frac{1}{1+e^{-\frac 1 q}} \rbra*{1 + e^{-\frac 1 q} } = 1.
    \end{align*}
    
    It can be verified that 
    \begin{equation}
    \label{eq:A-inv}
    \begin{aligned}
    A^{-1} = \rbra*{1+e^{-\frac 1 q}} \Bigg( \ket{0}\bra{1} \otimes \rbra*{I - e^{-\frac{1}{q}} P^{-1}}^{-1} + \ket{1}\bra{0} \otimes \rbra*{I - e^{-\frac{1}{q}} P}^{-1} \Bigg).
    \end{aligned}
    \end{equation}
    Then, 
    \begin{align*}
        \Abs*{A^{-1}} 
        & \leq \rbra*{1+e^{-\frac 1 q}} \max\Bigg\{ \Abs*{\rbra*{I - e^{-\frac{1}{q}} P^{-1}}^{-1}},
        \Abs*{\rbra*{I - e^{-\frac{1}{q}} P}^{-1}} \Bigg\} \\
        & \leq \rbra*{1+e^{-\frac 1 q}} \sum_{k=0}^{\infty} e^{-\frac k q} \\
        & = \frac{1+e^{-\frac 1 q}}{1-e^{-\frac 1 q}}.
    \end{align*}
    For every $\delta > 0$, it can be seen that for sufficiently large $q \geq 1$, we have
    \[
        \Abs*{A^{-1}} \leq \frac{1+e^{-\frac 1 q}}{1-e^{-\frac 1 q}} \leq \rbra{2+\delta} q.
    \]
    Therefore, $I/\kappa \leq A \leq I$ for $\kappa = \Abs{A} \Abs{A^{-1}} = \rbra{2+\delta}q$.
    \end{proof}
    
    Now we consider the solution to $\textup{QLSP}\rbra{6qN, \rbra{2+\delta}q, \varepsilon}$ with respect to $A$, which is
    \[
    \ket{\psi} = \frac{A^{-1}\ket{\bar{0}}}{\Abs{A^{-1}\ket{\bar{0}}}}, 
    \]
    where we denote $\ket{\bar{0}} = \ket{0}\ket{0}\ket{0}$ to avoid ambiguity and each $\ket{0}$ corresponds to one of the three subsystems of $A$ defined by \cref{eq:def-A}.
    After measuring $\ket{\psi}$ on the computational basis, the outcome satisfies the following properties. 
    
    \begin{lemma}
    \label{lemma:measure-psi}
        Let $\rbra{b, j, x}$ be the measurement outcome of $\ket{\psi}$ on the computational basis.
        Then, 
        \[
        \Pr\sbra{ q+1 \leq j \leq 2q } \geq \frac{e^{-4}-e^{-6}}{1-e^{-6}} > 0.015.
        \]
        Moreover, 
        \[
        x = \begin{cases}
            \Pi_j\rbra{0}, & 0 \leq j \leq q, \\
            \Pi_q\rbra{0}, & q+1 \leq j \leq 2q, \\
            \Pi_{3q+1-j}\rbra{0}, & 2q+1 \leq j < 3q.
        \end{cases}
        \]
        Here, $\Pi_j = \pi_j \circ \dots \circ \pi_2 \circ \pi_1$ for $0 \leq j \leq q$.
    \end{lemma}

    \begin{proof}
        It is easy to verify the relation between $x$ and $j$. Let $x_j$ be the value of $x$ corresponding to $j$ defined by
        \[
        x_j = \begin{cases}
            \Pi_j\rbra{0}, & 0 \leq j \leq q, \\
            \Pi_q\rbra{0}, & q+1 \leq j \leq 2q, \\
            \Pi_{3q+1-j}\rbra{0}, & 2q+1 \leq j < 3q.
        \end{cases}
        \]
        We only have to compute $\Pr\sbra{ q+1 \leq j \leq 2q }$ as follows. 
        Let 
        \[
        M = \sum_{j=q+1}^{2q} I \otimes \ket{j}\bra{j} \otimes I
        \]
        be the projector onto the subspace where the second register $j$ is between $q+1$ and $2q$ (inclusive). 
        Then,
        \[
        \Pr\sbra*{ q+1 \leq j \leq 2q } = \frac{\Abs{ M A^{-1} \ket{\bar{0}} }^2}{\Abs{A^{-1}\ket{\bar{0}}}^2}.
        \]
        To this end, we first compute $A^{-1}\ket{\bar{0}}$ by \cref{eq:A-inv}:
        \begin{align*}
        A^{-1} \ket{0}\ket{0}\ket{0} 
        & = \rbra*{1+e^{-\frac 1 q}} \Bigg( \ket{0}\bra{1} \otimes \rbra*{I - e^{-\frac{1}{q}} P^{-1}}^{-1} 
        + \ket{1}\bra{0} \otimes \rbra*{I - e^{-\frac{1}{q}} P}^{-1} \Bigg) \ket{0} \ket{0} \ket{0} \\
        & = \rbra*{1+e^{-\frac 1 q}} \ket{1} \otimes \rbra*{I - e^{-\frac{1}{q}} P}^{-1} \ket{0} \ket{0} \\
        & = \rbra*{1+e^{-\frac 1 q}} \ket{1} \otimes \sum_{k=0}^{\infty} e^{-\frac{k}{q}} P^{k} \ket{0} \ket{0} \\
        & = \rbra*{1+e^{-\frac 1 q}} \ket{1} \otimes \sum_{k=0}^{\infty} e^{-\frac{k}{q}} \ket{k \bmod 3q} \ket{x_{k \bmod 3q}} \\
        & = \rbra*{1+e^{-\frac 1 q}} \ket{1} \otimes \sum_{j=0}^{3q-1} \sum_{\ell=0}^{\infty} e^{-\frac{j + 3q\ell}{q}} \ket{j} \ket{x_{j}} \\
        & = \rbra*{1+e^{-\frac 1 q}} \ket{1} \otimes \sum_{j=0}^{3q-1} \frac{e^{-\frac j q}}{1 - e^{-3}} \ket{j} \ket{x_{j}}.
        \end{align*}
    With this, we have
    \begin{align*}
        \Abs*{A^{-1}\ket{\bar{0}}}^2 & = \sum_{j=0}^{3q-1} \rbra*{\frac{\rbra*{1+e^{-\frac 1 q}}e^{-\frac j q}}{1 - e^{-3}}}^2 \\
        & = \rbra*{\frac{1+e^{-\frac 1 q}}{1 - e^{-3}}}^2 \frac{1 - e^{-6}}{1 - e^{-\frac 2 q}},
    \end{align*}
    \begin{align*}
    \Abs*{MA^{-1}\ket{\bar{0}}}^2 
    & = \sum_{j=q+1}^{2q} \rbra*{\frac{\rbra*{1+e^{-\frac 1 q}}e^{-\frac j q}}{1 - e^{-3}}}^2 \\
    & = \rbra*{\frac{1+e^{-\frac 1 q}}{1 - e^{-3}}}^2 \frac{e^{-\frac {2\rbra{q+1}} q} \rbra*{1 - e^{-2}}}{1 - e^{-\frac 2 q}}.
    \end{align*}
    Then, 
    \[
    \frac{\Abs{ M A^{-1} \ket{\bar{0}} }^2}{\Abs{A^{-1}\ket{\bar{0}}}^2}
    = \frac{e^{-2} - e^{-4}}{1 - e^{-6}} e^{-\frac 2 q} \geq \frac{e^{-4} - e^{-6}}{1 - e^{-6}}. 
    \]
    \end{proof}

    By \cref{lemma:measure-psi}, we can obtain the solution $\Pi_q\rbra{0}$ to $\textup{PermChain}\rbra{N, q}$ with a constant probability by measuring $\ket{\psi}$ on the computational basis with outcome $\rbra{b, j, x}$: it holds with probability $> 0.015$ that $q+1 \leq j \leq 2q$ and thus $x = x_j = \Pi_q\rbra{0}$.

    \subsection{Constructing quantum query oracles}
    \label{sec:construct-access}

    As shown in \cref{lemma:basic-A}, $A$ is a $2$-sparse $6qN \times 6qN$ Hermitian matrix.
    In this subsection, we will explicitly construct the quantum oracles $\mathcal{O}_s$ and $\mathcal{O}_A$ for quantum query access to $A$.

    \begin{lemma}
    \label{lemma:OsOabyOpi}
        Let $\mathcal{O}_s$ and $\mathcal{O}_A$ be quantum oracles, defined by \cref{eq:def-Os} and \cref{eq:def-OA}, respectively, for sparse matrix $A$ defined by \cref{eq:def-A}.
        Then, each of $\mathcal{O}_s$ and $\mathcal{O}_A$ can be implemented using $1$ query to $\mathcal{O}_{\pi}$ defined by \cref{eq:def-Opi}.
    \end{lemma}
    \begin{proof}
        The matrix representation of $A$ by \cref{eq:def-A} is 
        \[
        A = \frac{1}{1+e^{-\frac 1 q}} \begin{bmatrix}
            0 & I - e^{-\frac 1 q} P \\
            I - e^{-\frac 1 q} P^{-1} & 0
        \end{bmatrix}.
        \]
        We first consider how to compute the entries in $I - e^{-\frac 1 q} P$, and the case for $I - e^{-\frac 1 q} P^{-1}$ is similar. 
        For $0 \leq j < 3q$ and $0 \leq x < N$, we have
        \[
        P \ket{j}\ket{x} = \begin{cases}
            \ket{\rbra{j+1} \bmod 3q} \ket{\pi_{j+1}\rbra{x}}, & 0 \leq j < q, \\
            \ket{\rbra{j+1} \bmod 3q} \ket{x}, & q \leq j < 2q, \\ 
            \ket{\rbra{j+1} \bmod 3q} \ket{\pi^{-1}_{3q-j}\rbra{x}}, & 2q \leq j < 3q. 
        \end{cases}
        \]
        By noting that $A \ket{1} \ket{j} \ket{x} = \ket{0} \otimes \rbra{I - e^{-\frac 1 q} P} \ket{j}\ket{x}$, we have
        \begin{equation}
        \label{eq:I-eP}
        \begin{aligned}
            A \ket{1} \ket{j} \ket{x} =
            \begin{cases}
            \ket{0} \ket{j}\ket{x} - e^{-\frac 1 q} \ket{0} \ket{\rbra{j+1} \bmod 3q} \ket{\pi_{j+1}\rbra{x}}, & 0 \leq j < q, \\
            \ket{0}\ket{j}\ket{x} - e^{-\frac 1 q} \ket{0}\ket{\rbra{j+1} \bmod 3q} \ket{x}, & q \leq j < 2q, \\ 
            \ket{0}\ket{j}\ket{x} - e^{-\frac 1 q} \ket{0}\ket{\rbra{j+1} \bmod 3q} \ket{\pi^{-1}_{3q-j}\rbra{x}}, & 2q \leq j < 3q. 
        \end{cases}
        \end{aligned}
        \end{equation}
        Similarly, we have
        \begin{equation}
        \label{eq:I-ePinv}
        \begin{aligned}
            A \ket{0} \ket{j} \ket{x} = 
            \begin{cases}
            \ket{1} \ket{j}\ket{x} - e^{-\frac 1 q} \ket{1} \ket{\rbra{j+1} \bmod 3q} \ket{\pi_{j+1}^{-1}\rbra{x}}, & 0 \leq j < q, \\
            \ket{1}\ket{j}\ket{x} - e^{-\frac 1 q} \ket{1}\ket{\rbra{j+1} \bmod 3q} \ket{x}, & q \leq j < 2q, \\ 
            \ket{1}\ket{j}\ket{x} - e^{-\frac 1 q} \ket{1}\ket{\rbra{j+1} \bmod 3q} \ket{\pi_{3q-j}\rbra{x}}, & 2q \leq j < 3q. 
            \end{cases}
        \end{aligned}
        \end{equation}
        We can use \cref{eq:I-eP} and \cref{eq:I-ePinv} to find the non-zero entries of $A$, and thus we can construct each of the quantum oracles $\mathcal{O}_s$ and $\mathcal{O}_A$ for $A$ using only $1$ query to $\mathcal{O}_\pi$. 
    \end{proof}

    \subsection{Reducing permutation chain to QLSP}

    \label{sec:reduction}

    Let $\mathcal{A}$ be a quantum query algorithm with time complexity $\mathsf{T}\rbra{\mathcal{A}} = T\rbra{N, \kappa, \varepsilon} =  \poly\rbra{\log\rbra{N}, \kappa, 1/\varepsilon}$, and with quantum query-depth complexity $\mathsf{Q}_{\parallel}\rbra{\mathcal{A}} = Q_{\parallel}^{\textup{sparse}}\rbra{N, \kappa, \varepsilon}$. 
    In other words, the algorithm $\mathcal{A}$ uses $Q_{\parallel}^{\textup{sparse}}\rbra{N, \kappa, \varepsilon}$ $k$-parallel queries to $\mathcal{O}_s$ and $\mathcal{O}_A$, where $k \leq T\rbra{N, \kappa, \varepsilon}$.

    Let $\varepsilon \in \rbra{0, \frac{e^{-4}-e^{-6}}{1-e^{-6}}}$ be a constant, e.g., $\varepsilon = 0.001 = \Theta\rbra{1}$. 
    Now let $\delta > 0$ be a real number.
    Let $n = q \geq 1$ be a sufficiently large integer such that $\rbra{2+\delta}^{-1}q^{-1}I \leq A \leq I$, whose existence is guaranteed by \cref{lemma:basic-A}. 
    Let $N = 2^n$ and $k = T\rbra{6qN, \rbra{2+\delta}q, \varepsilon} = \poly\rbra{n}$. 
    Let $\pi_1, \pi_2, \dots, \pi_q$ be permutations chosen from $S_N$ uniformly at random. 
    By \cref{thm:permchain}, for every $1 \leq t \leq \rbra{q-1}/2$ and every quantum algorithm $\mathcal{A}$ that uses $t$ $k$-parallel queries to $\mathcal{O}_{\pi}$, we have
    \begin{equation}
    \label{eq:expected-succ-prob}
    \begin{aligned}
        \E_{\pi_1, \pi_2, \dots, \pi_q} \sbra*{ \Pr\sbra*{\mathcal{A} \text{ outputs } \Pi_q\rbra{0}} } = O\rbra*{q \sqrt{\frac k N}} \leq \frac{\poly\rbra{n}}{2^n} = \mathsf{negl}\rbra{n},
    \end{aligned}
    \end{equation}
    where $\mathsf{negl}\rbra{n}$ means a negligible function of $n$. 
    
    Using the quantum oracles $\mathcal{O}_s$ and $\mathcal{O}_A$ constructed by $\mathcal{O}_{\pi}$ in \cref{lemma:OsOabyOpi}, the algorithm $\mathcal{A}$ can produce a quantum state $\ket{\tilde \psi}$ such that 
    \[
    \Abs*{ \ket{\psi} - \ket{\tilde \psi} } \leq \varepsilon,
    \]
    using $t = Q_{\parallel}^{\textup{sparse}}\rbra{6qN, \rbra{2+\delta}q, \varepsilon}$ $k$-parallel queries to $\mathcal{O}_s$ and $\mathcal{O}_A$. 
    Then, by \cref{lemma:measure-psi}, with probability at least $\frac{e^{-4}-e^{-6}}{1-e^{-6}} - \varepsilon = \Omega\rbra{1}$, we can find $\Pi_q\rbra{0}$ by measuring $\ket{\tilde \psi}$ on the computational basis, i.e.,
    \begin{equation}
    \label{eq:succ-prob}
    \Pr\sbra*{\mathcal{A} \text{ outputs } \Pi_q\rbra{0}} = \Omega\rbra{1}.
    \end{equation}
    Now if $t \leq \rbra{q-1}/2$, then $\mathcal{A}$ can be considered as a quantum query algorithm that makes $t$ $k$-parallel queries to $\mathcal{O}_\pi$. Then, by \cref{eq:expected-succ-prob}, we have
    \[
    \E_{\pi_1, \pi_2, \dots, \pi_q} \sbra*{ \Pr\sbra*{\mathcal{A} \text{ outputs } \Pi_q\rbra{0}} } = \mathsf{negl}\rbra{n},
    \]
    which leads to a contradiction since \cref{eq:succ-prob} does not depend on the choices of $\pi_1, \pi_2, \dots, \pi_q$.
    Therefore, it must be the case that $t > \rbra{q-1}/2$, and thus $t \geq q/2$ (note that $t$ and $q$ are positive integers).
    This means that 
    \[
    Q_{\parallel}^{\textup{sparse}}\rbra{6qN, \rbra{2+\delta}q, \varepsilon} \geq \frac q 2.
    \]
    By letting $\kappa = \rbra{2+\delta}n$, we have 
    \[
    Q_{\parallel}^{\textup{sparse}}\rbra*{\frac{6}{2+\delta} \kappa 2^{\frac{1}{2+\delta} \kappa}, \kappa, \varepsilon} \geq \frac{1}{2\rbra{2+\delta}}\kappa,
    \]
    which is
    \[
    Q_{\parallel}^{\textup{sparse}}\rbra{3 \kappa 2^{\frac{\kappa}{2}}, \kappa, \varepsilon} \geq \frac{1}{4+2\delta} \kappa.
    \]
    Because of the arbitrariness of $\delta$, these yield the proof of \cref{thm:main}. 

    \begin{remark} \label{rmk:kappa}
        The choice of $N$ can be loosened to $N = 2^{n^c}$ for any constant $c > 0$, where the above proof chooses $c = 1$ for simplicity. 
        With such choice of $N$, 
        \cref{thm:main} can be strengthened to 
        \[
        \limsup_{\kappa \to +\infty} \frac{Q_{\parallel}^{\textup{sparse}}\rbra{3\kappa 2^{\rbra{\frac{\kappa}{2}}^c}, \kappa, \varepsilon}}{\kappa} \geq \frac 1 4.
        \]
    \end{remark}

    \section{Quantum Depth Lower Bound for Block-QLSP}
    
    \label{sec:block-QLSP}

    In this section, we extend the proof of \cref{thm:main} to the case of block-QLSP. 
    The idea is to solve the specific instance of sparse-QLSP in the proof of \cref{thm:main} by quantum algorithms for block-QLSP.
    The difference is that due to imperfect implementations of the block-encoding of sparse matrices (caused by \cref{lemma:sparse-to-block}), we have to analyze the perturbation of linear systems (see \cref{lemma:perturb-QLSP}) in order to ensure the success probability of the quantum algorithms.

    The quantum depth lower bound for block-QLSP is formally stated as follows. 

    \begin{theorem}
    \label{thm:block-QLSP}
        For every constant $0 < \varepsilon < \frac{e^{-4}-e^{-6}}{1-e^{-6}} \approx 0.015$, we have
        \[
        \limsup_{\kappa \to +\infty} \frac{Q_{\parallel}^{\textup{block}}\rbra{3\kappa 2^{\frac{\kappa}{4}-1}, \kappa, \varepsilon}}{\kappa} \geq \frac 1 {32}.
        \]
    \end{theorem}

    \begin{proof}
        In the proof of \cref{thm:main}, for every $\delta > 0$, we construct, for every sufficiently large $n \geq 1$, a $2$-sparse Hermitian matrix $A \in \mathbb{C}^{6qN \times 6qN}$ with $q = n$ and $N = 2^n$ such that $I/\kappa \leq A \leq I$, where $\kappa = \rbra{2+\delta}q$. 

        Now suppose that we are given quantum oracles $\mathcal{O}_s$ and $\mathcal{O}_A$ for sparse matrix $A$. 
        We will reduce sparse-QLSP to block-QLSP by constructing the quantum oracle $U_A$ that is an approximate block-encoding of $A$ using queries to $\mathcal{O}_s$ and $\mathcal{O}_A$. 
        Let $0 < \varepsilon' < 1/\kappa$ to be determined later. 
        By \cref{lemma:sparse-to-block}, we can implement a unitary operator $U_A$ such that $U_A$ is a $\rbra{2, a, \varepsilon'}$-block-encoding of $A$, using $2$ queries to $\mathcal{O}_s$, $2$ queries to $\mathcal{O}_A$, and $O\rbra{\log\rbra{N}+\log^{2.5}\rbra{1/\varepsilon'}}$ one- and two-qubit quantum gates, where $a = \ceil{\log_2\rbra{N}}+3$.
        In other words, $U_A$ is a $\rbra{1, a, 0}$-block-encoding of $A'$ such that $\Abs{2A' - A} \leq \varepsilon'$.
        It can be shown that $I/\kappa' \leq A' \leq I$, where $\kappa' = 2\kappa/\rbra{1-2\kappa\varepsilon'}$.
        Considering the perturbation $\Abs{A/2-A'} \leq \varepsilon'/2$ of linear systems (see \cref{lemma:perturb-QLSP}), we have
        \begin{align*}
            \Abs*{ \frac{A^{-1}\ket{0}}{\Abs{A^{-1}\ket{0}}} - \frac{A'^{-1}\ket{0}}{\Abs{A'^{-1}\ket{0}}} } 
            & \leq \frac{\rbra{2\kappa}^2 \cdot \rbra{2\kappa+1} \cdot \frac{\varepsilon'}{2}}{1 - \rbra{2\kappa} \cdot \frac{\varepsilon'}{2}} \\
            & = \frac{2\kappa^2\rbra{2\kappa+1}\varepsilon'}{1 - \kappa\varepsilon'}.
        \end{align*}
        For $\kappa \geq 4$, we choose $\varepsilon' = 1/2\kappa^5 = 1/\poly\rbra{\kappa} = 1/\poly\rbra{n}$, then it can be shown that 
        \begin{equation}
        \label{eq:perturb-Aprime}
        \Abs*{ \frac{A^{-1}\ket{0}}{\Abs{A^{-1}\ket{0}}} - \frac{A'^{-1}\ket{0}}{\Abs{A'^{-1}\ket{0}}} } \leq \frac{2\kappa^2\rbra{2\kappa+1}\varepsilon'}{1 - \kappa\varepsilon'} < \frac{1}{\kappa}.
        \end{equation}

        Suppose that there is a quantum query algorithm $\mathcal{A}$ for $\textup{QLSP}\rbra{N, \kappa, \varepsilon}$ in the block-encoded input model such that it has quantum query-depth complexity $\mathsf{Q}_{\parallel}\rbra{\mathcal{A}} = Q_{\parallel}^{\textup{block}}\rbra{N, \kappa, \varepsilon}$ and quantum time complexity $\mathsf{T}\rbra{\mathcal{A}} = T\rbra{N, \kappa, \varepsilon} =  \poly\rbra{\log\rbra{N}, \kappa, 1/\varepsilon}$.
        Now for every constant $\varepsilon \in \rbra{0, \frac{e^{-4}-e^{-6}}{1-e^{-6}}}$, applying $\mathcal{A}$ with quantum oracle $U_A$ to solve $\textup{QLSP}\rbra{6qN, \kappa', \varepsilon}$ for $A'$, we obtain a quantum state $\ket{\tilde \psi}$ (with high probability) such that
        \begin{equation}
        \label{eq:err-QLSP}
        \Abs*{\ket{\tilde \psi} - \frac{A'^{-1}\ket{0}}{\Abs{A'^{-1}\ket{0}}}} \leq \varepsilon,
        \end{equation}
        using $Q_{\parallel}^{\textup{block}}\rbra{6qN, \kappa', \varepsilon}$ queries to $U_A$ and the number of one- and two-qubit gates is $T\rbra{6qN, \kappa', \varepsilon} + Q_{\parallel}^{\textup{block}}\rbra{6qN, \kappa', \varepsilon} \cdot O\rbra{\log\rbra{N} + \log^{2.5}\rbra{1/\varepsilon'}} = \poly\rbra{\log\rbra{6qN}, \kappa', 1/\varepsilon} = \poly\rbra{n}$. 

        By \cref{eq:perturb-Aprime} and \cref{eq:err-QLSP}, we have that for $\kappa \geq 4$, 
        \[
        \Abs*{\ket{\tilde \psi} - \frac{A^{-1}\ket{0}}{\Abs{A^{-1}\ket{0}}}} \leq \frac{1}{\kappa} + \varepsilon.
        \]
        Now we consider the permutation chain problem considered in \cref{sec:reduction}.
        By \cref{lemma:measure-psi}, with probability at least $\frac{e^{-4}-e^{-6}}{1-e^{-6}} - \varepsilon - 1/\kappa = \Omega\rbra{1}$ (for sufficiently large $\kappa$), we can find $\Pi_q\rbra{0}$ by measuring $\ket{\tilde \psi}$ on the computational basis. 
        Here, we note that $U_A$ can be implemented by $2$ queries to $\mathcal{O}_s$ and $\mathcal{O}_A$; by \cref{lemma:OsOabyOpi}, $U_A$ can be then implemented by $4$ queries to $\mathcal{O}_\pi$.
        Therefore, we can obtain $\ket{\tilde \psi}$ using $4Q_{\parallel}^{\textup{block}}\rbra{6qN, \kappa', \varepsilon}$ parallel queries to $\mathcal{O}_\pi$. 
        An argument similar to \cref{sec:reduction} will also show that this is impossible unless 
        \[
        4Q_{\parallel}^{\textup{block}}\rbra{6qN, \kappa', \varepsilon} \geq \frac q 2,
        \]
        where
        \[
        \kappa' = \frac{2\kappa}{1-2\kappa\varepsilon'} = \frac{2\kappa}{1-\kappa^{-4}}.
        \]
        By letting $\kappa = \rbra{2+\delta}n$, we have
        \[
        Q_{\parallel}^{\textup{block}}\rbra*{6 \cdot \frac{\kappa}{2+\delta} \cdot 2^{\frac{\kappa}{2+\delta}}, \frac{2\kappa}{1-\kappa^{-4}}, \varepsilon} \geq \frac{\kappa}{8\rbra{2+\delta}},
        \]
        which implies
        \[
        Q_{\parallel}^{\textup{block}}\rbra*{3\kappa 2^{\frac{\kappa}{4}-1}, \frac{\kappa}{1-16\kappa^{-4}}, \varepsilon} \geq \frac{\kappa}{16\rbra{2+\delta}}.
        \]
        Because $\delta$ can be chosen arbitrarily large and the term $1-16\kappa^{-4}$ tends to $1$ when $\kappa$ is large, these yield the proof. 
    \end{proof}

    \section*{Acknowledgment}

    The authors would like to thank Robin Kothari for communication regarding the related work \cite{HK21}.
    They also thank Fran{\c{c}}ois Le Gall for helpful comments on an early version of this paper.

    Qisheng Wang was supported by the MEXT Quantum Leap Flagship Program (MEXT Q-LEAP) grants No.~JPMXS0120319794. Zhicheng Zhang was supported by the Sydney Quantum Academy, NSW, Australia.

    \addcontentsline{toc}{section}{References}
    
    \bibliographystyle{unsrturl}
    \bibliography{main}

\begin{thebibliography}{10}

\bibitem{HHL09}
Aram~W. Harrow, Avinatan Hassidim, and Seth Lloyd.
\newblock Quantum algorithm for linear systems of equations.
\newblock {\em Physical Review Letters}, 103(15):150502, 2009.
\newblock \href {https://doi.org/10.1103/PhysRevLett.103.150502}
  {\path{doi:10.1103/PhysRevLett.103.150502}}.

\bibitem{BWP+17}
Jacob Biamonte, Peter Wittek, Nicola Pancotti, Patrick Rebentrost, Nathan
  Wiebe, and Seth Lloyd.
\newblock Quantum machine learning.
\newblock {\em Nature}, 549(7671):195--202, 2017.
\newblock \href {https://doi.org/10.1038/nature23474}
  {\path{doi:10.1038/nature23474}}.

\bibitem{CRO+19}
Yudong Cao, Jonathan Romero, Jonathan~P. Olson, Matthias Degroote, Peter~D.
  Johnson, Mária Kieferová, Ian~D. Kivlichan, Tim Menke, Borja Peropadre,
  Nicolas P.~D. Sawaya, Sukin Sim, Libor Veis, and Alán Aspuru-Guzik.
\newblock Quantum chemistry in the age of quantum computing.
\newblock {\em Chemical Reviews}, 119(19):10856--10915, 2019.
\newblock \href {https://doi.org/10.1021/acs.chemrev.8b00803}
  {\path{doi:10.1021/acs.chemrev.8b00803}}.

\bibitem{OML19}
Román Orús, Samuel Mugel, and Enrique Lizaso.
\newblock Quantum computing for finance: overview and prospects.
\newblock {\em Reviews in Physics}, 4:100028, 2019.
\newblock \href {https://doi.org/10.1016/j.revip.2019.100028}
  {\path{doi:10.1016/j.revip.2019.100028}}.

\bibitem{GSLW19}
Andr{\'{a}}s Gily{\'{e}}n, Yuan Su, Guang~Hao Low, and Nathan Wiebe.
\newblock Quantum singular value transformation and beyond: exponential
  improvements for quantum matrix arithmetics.
\newblock In {\em Proceedings of the 51st Annual ACM SIGACT Symposium on Theory
  of Computing}, pages 193--204, 2019.
\newblock \href {https://doi.org/10.1145/3313276.3316366}
  {\path{doi:10.1145/3313276.3316366}}.

\bibitem{LP96}
A.~Luis and J.~Perina.
\newblock Optimum phase-shift estimation and the quantum description of the
  phase difference.
\newblock {\em Physical Review A}, 54(5):4564, 1996.
\newblock \href {https://doi.org/10.1103/PhysRevA.54.4564}
  {\path{doi:10.1103/PhysRevA.54.4564}}.

\bibitem{CEMM98}
Richard Cleve, Artur Ekert, Chiara Macchiavello, and Michele Mosca.
\newblock Quantum algorithms revisited.
\newblock {\em Proceedings of the Royal Society of A: Mathematical, Physical
  and Engineering Sciences}, 454(1969):339--354, 1998.
\newblock \href {https://doi.org/10.1098/rspa.1998.0164}
  {\path{doi:10.1098/rspa.1998.0164}}.

\bibitem{BDM99}
V.~Buzek, R.~Derka, and S.~Massar.
\newblock Optimal quantum clocks.
\newblock {\em Physical Review Letters}, 82(10):2207, 1999.
\newblock \href {https://doi.org/10.1103/PhysRevLett.82.2207}
  {\path{doi:10.1103/PhysRevLett.82.2207}}.

\bibitem{Amb12}
Andris Ambainis.
\newblock Variable time amplitude amplification and quantum algorithms for
  linear algebra problems.
\newblock In {\em Proceedings of the 29th International Symposium on
  Theoretical Aspects of Computer Science}, pages 636--647, 2012.
\newblock \href {https://doi.org/10.4230/LIPIcs.STACS.2012.636}
  {\path{doi:10.4230/LIPIcs.STACS.2012.636}}.

\bibitem{CKS17}
Andrew~M. Childs, Robin Kothari, and Rolando~D. Somma.
\newblock Quantum algorithm for systems of linear equations with exponentially
  improved dependence on precision.
\newblock {\em SIAM Journal on Computing}, 46(6):1920--1950, 2017.
\newblock \href {https://doi.org/10.1137/16M1087072}
  {\path{doi:10.1137/16M1087072}}.

\bibitem{CW12}
Andrew~M. Childs and Nathan Wiebe.
\newblock Hamiltonian simulation using linear combinations of unitary
  operations.
\newblock {\em Quantum Information and Computation}, 12(11--12):901--924, 2012.
\newblock \href {https://doi.org/10.26421/QIC12.11-12-1}
  {\path{doi:10.26421/QIC12.11-12-1}}.

\bibitem{BCC+15}
Dominic~W. Berry, Andrew~M. Childs, Richard Cleve, Robin Kothari, and
  Rolando~D. Somma.
\newblock Simulating {Hamiltonian} dynamics with a truncated {Taylor} series.
\newblock {\em Physical Review Letters}, 114(9):090502, 2015.
\newblock \href {https://doi.org/10.1103/PhysRevLett.114.090502}
  {\path{doi:10.1103/PhysRevLett.114.090502}}.

\bibitem{SSO19}
Yiğit Subasi, Rolando~D. Somma, and Davide Orsucci.
\newblock Quantum algorithms for systems of linear equations inspired by
  adiabatic quantum computing.
\newblock {\em Physical Review Letters}, 122(6):060504, 2019.
\newblock \href {https://doi.org/10.1103/PhysRevLett.122.060504}
  {\path{doi:10.1103/PhysRevLett.122.060504}}.

\bibitem{AL22}
Dong An and Lin Lin.
\newblock Quantum linear system solver based on time-optimal adiabatic quantum
  computing and quantum approximate optimization algorithm.
\newblock {\em ACM Transactions on Quantum Computing}, 3(2):5:1--28, 2022.
\newblock \href {https://doi.org/10.1145/3498331} {\path{doi:10.1145/3498331}}.

\bibitem{LT20}
Lin Lin and Yu~Tong.
\newblock Optimal polynomial based quantum eigenstate filtering with
  application to solving quantum linear systems.
\newblock {\em Quantum}, 4:361, 2020.
\newblock \href {https://doi.org/10.22331/q-2020-11-11-361}
  {\path{doi:10.22331/q-2020-11-11-361}}.

\bibitem{OD21}
Davide Orsucci and Vedran Dunjko.
\newblock On solving classes of positive-definite quantum linear systems with
  quadratically improved runtime in the condition number.
\newblock {\em Quantum}, 5:573, 2021.
\newblock \href {https://doi.org/10.22331/q-2021-11-08-573}
  {\path{doi:10.22331/q-2021-11-08-573}}.

\bibitem{CAS+22}
Pedro C.~S. Costa, Dong An, Yuval~R. Sanders, Yuan Su, Ryan Babbush, and
  Dominic~W. Berry.
\newblock Optimal scaling quantum linear-systems solver via discrete adiabatic
  theorem.
\newblock {\em PRX Quantum}, 3(4):040303, 2022.
\newblock \href {https://doi.org/10.1103/PRXQuantum.3.040303}
  {\path{doi:10.1103/PRXQuantum.3.040303}}.

\bibitem{HK21}
Aram~W. Harrow and Robin Kothari.
\newblock In preparation, 2021.

\bibitem{CW00}
Richard Cleve and John Watrous.
\newblock Fast parallel circuits for the quantum {Fourier} transform.
\newblock In {\em Proceedings of the 41st Annual Symposium on Foundations of
  Computer Science}, pages 526--536, 2000.
\newblock \href {https://doi.org/10.1109/SFCS.2000.892140}
  {\path{doi:10.1109/SFCS.2000.892140}}.

\bibitem{Sho94}
Peter~W. Shor.
\newblock Algorithms for quantum computation: discrete logarithms and
  factoring.
\newblock In {\em Proceedings of the 35th Annual Symposium on Foundations of
  Computer Science}, pages 124--134, 1994.
\newblock \href {https://doi.org/10.1109/SFCS.1994.365700}
  {\path{doi:10.1109/SFCS.1994.365700}}.

\bibitem{PS13}
Paul Pham and Krysta~M. Svore.
\newblock A {2D} nearest-neighbor quantum architecture for factoring in
  polylogarithmic depth.
\newblock {\em Quantum Information and Computation}, 13(11--12):937--962, 2013.
\newblock \href {https://doi.org/10.26421/QIC13.11-12-3}
  {\path{doi:10.26421/QIC13.11-12-3}}.

\bibitem{RS14}
Martin Rötteler and Rainer Steinwandt.
\newblock A quantum circuit to find discrete logarithms on ordinary binary
  elliptic curves in depth {$O(\log^2n)$}.
\newblock {\em Quantum Information and Computation}, 14(9--10):888--900, 2014.
\newblock \href {https://doi.org/10.26421/QIC14.9-10-11}
  {\path{doi:10.26421/QIC14.9-10-11}}.

\bibitem{Gro96}
Lov~K. Grover.
\newblock A fast quantum mechanical algorithm for database search.
\newblock In {\em Proceedings of the 28th Annual ACM Symposium on Theory of
  Computing}, pages 212--219, 1996.
\newblock \href {https://doi.org/10.1145/237814.237866}
  {\path{doi:10.1145/237814.237866}}.

\bibitem{Zal99}
Christof Zalka.
\newblock Grover’s quantum searching algorithm is optimal.
\newblock {\em Physical Review A}, 60(4):2746, 1999.
\newblock \href {https://doi.org/10.1103/PhysRevA.60.2746}
  {\path{doi:10.1103/PhysRevA.60.2746}}.

\bibitem{GWC00}
Robert~M. Gingrich, Colin~P. Williams, and Nicolas~J. Cerf.
\newblock Generalized quantum search with parallelism.
\newblock {\em Physical Review A}, 61(5):052313, 2000.
\newblock \href {https://doi.org/10.1103/PhysRevA.61.052313}
  {\path{doi:10.1103/PhysRevA.61.052313}}.

\bibitem{GR04}
Lov~K. Grover and Jaikumar Radhakrishnan.
\newblock Quantum search for multiple items using parallel queries.
\newblock ArXiv e-prints, 2004.
\newblock \href {http://arxiv.org/abs/quant-ph/0407217}
  {\path{arXiv:quant-ph/0407217}}.

\bibitem{Bur19}
Paul Burchard.
\newblock Lower bounds for parallel quantum counting.
\newblock ArXiv e-prints, 2019.
\newblock \href {http://arxiv.org/abs/1910.04555} {\path{arXiv:1910.04555}}.

\bibitem{GTKL+22}
Tudor Giurgica-Tiron, Iordanis Kerenidis, Farrokh Labib, Anupam Prakash, and
  William Zeng.
\newblock Low depth algorithms for quantum amplitude estimation.
\newblock {\em Quantum}, 6:745, 2022.
\newblock \href {https://doi.org/10.22331/q-2022-06-27-745}
  {\path{doi:10.22331/q-2022-06-27-745}}.

\bibitem{Amb07}
Andris Ambainis.
\newblock Quantum walk algorithm for element distinctness.
\newblock {\em SIAM Journal on Computing}, 37(1):210--239, 2007.
\newblock \href {https://doi.org/10.1137/S0097539705447311}
  {\path{doi:10.1137/S0097539705447311}}.

\bibitem{JMdW17}
Stacey Jeffery, Frederic Magniez, and Ronald de~Wolf.
\newblock Optimal parallel quantum query algorithms.
\newblock {\em Algorithmica}, 79:509--529, 2017.
\newblock \href {https://doi.org/10.1007/s00453-016-0206-z}
  {\path{doi:10.1007/s00453-016-0206-z}}.

\bibitem{BGK18}
Sergey Bravyi, David Gosset, and Robert König.
\newblock Quantum advantage with shallow circuits.
\newblock {\em Science}, 362(6412):308--311, 2018.
\newblock \href {https://doi.org/10.1126/science.aar3106}
  {\path{doi:10.1126/science.aar3106}}.

\bibitem{BWKST19}
Adam Bene~Watts, Robin Kothari, Luke Schaeffer, and Avishay Tal.
\newblock Exponential separation between shallow quantum circuits and unbounded
  fan-in shallow classical circuits.
\newblock In {\em Proceedings of the 51st Annual ACM SIGACT Symposium on Theory
  of Computing}, pages 515--526, 2019.
\newblock \href {https://doi.org/10.1145/3313276.3316404}
  {\path{doi:10.1145/3313276.3316404}}.

\bibitem{LG19}
François Le~Gall.
\newblock Average-case quantum advantage with shallow circuits.
\newblock In {\em Proceedings of the 34th Computational Complexity Conference},
  pages 21:1--21:20, 2019.
\newblock \href {https://doi.org/10.4230/LIPIcs.CCC.2019.21}
  {\path{doi:10.4230/LIPIcs.CCC.2019.21}}.

\bibitem{BGKT20}
Sergey Bravyi, David Gosset, Robert König, and Marco Tomamichel.
\newblock Quantum advantage with noisy shallow circuits.
\newblock {\em Nature Physics}, 16(10):1040--1045, 2020.
\newblock \href {https://doi.org/10.1038/s41567-020-0948-z}
  {\path{doi:10.1038/s41567-020-0948-z}}.

\bibitem{GS20}
Daniel Grier and Luke Schaeffer.
\newblock Interactive shallow {Clifford} circuits: quantum advantage against
  {$\mathsf{NC}^1$} and beyond.
\newblock In {\em Proceedings of the 52nd Annual ACM SIGACT Symposium on Theory
  of Computing}, pages 875--888, 2020.
\newblock \href {https://doi.org/10.1145/3357713.3384332}
  {\path{doi:10.1145/3357713.3384332}}.

\bibitem{CSV21}
Matthew Coudron, Jalex Stark, and Thomas Vidick.
\newblock Trading locality for time: certifiable randomness from low-depth
  circuits.
\newblock {\em Communications in Mathematical Physics}, 382:49--86, 2021.
\newblock \href {https://doi.org/10.1007/s00220-021-03963-w}
  {\path{doi:10.1007/s00220-021-03963-w}}.

\bibitem{BWP23}
Adam Bene~Watts and Natalie Parham.
\newblock Unconditional quantum advantage for sampling with shallow circuits.
\newblock ArXiv e-prints, 2023.
\newblock \href {http://arxiv.org/abs/2301.00995} {\path{arXiv:2301.00995}}.

\bibitem{QKW22}
Yihui Quek, Eneet Kaur, and Mark~M. Wilde.
\newblock Multivariate trace estimation in constant quantum depth.
\newblock {\em Quantum}, 8:1220, 2024.
\newblock \href {https://doi.org/10.22331/q-2024-01-10-1220}
  {\path{doi:10.22331/q-2024-01-10-1220}}.

\bibitem{ZWY21}
Zhicheng Zhang, Qisheng Wang, and Mingsheng Ying.
\newblock Parallel quantum algorithm for {Hamiltonian} simulation.
\newblock {\em Quantum}, 8:1228, 2024.
\newblock \href {https://doi.org/10.22331/q-2024-01-15-1228}
  {\path{doi:10.22331/q-2024-01-15-1228}}.

\bibitem{CCH+23}
Nai-Hui Chia, Kai-Min Chung, Yao-Ching Hsieh, Han-Hsuan Lin, Yao-Ting Lin, and
  Yu-Ching Shen.
\newblock On the impossibility of general parallel fast-forwarding of
  {Hamiltonian} simulation.
\newblock In {\em Proceedings of the 38th Computational Complexity Conference},
  pages 33:1--33:45, 2023.
\newblock \href {https://doi.org/10.4230/LIPIcs.CCC.2023.33}
  {\path{doi:10.4230/LIPIcs.CCC.2023.33}}.

\bibitem{JST+20}
Jiaqing Jiang, Xiaoming Sun, Shang-Hua Teng, Bujiao Wu, Kewen Wu, and Jialin
  Zhang.
\newblock Optimal space-depth trade-off of {CNOT} circuits in quantum logic
  synthesis.
\newblock In {\em Proceedings of the 2020 ACM-SIAM Symposium on Discrete
  Algorithms}, pages 213--229, 2020.
\newblock \href {https://doi.org/10.1137/1.9781611975994.13}
  {\path{doi:10.1137/1.9781611975994.13}}.

\bibitem{STY+23}
Xiaoming Sun, Guojing Tian, Shuai Yang, Pei Yuan, and Shengyu Zhang.
\newblock Asymptotically optimal circuit depth for quantum state preparation
  and general unitary synthesis.
\newblock {\em IEEE Transactions on Computer-Aided Design of Integrated
  Circuits and Systems}, to appear, 2023.
\newblock \href {https://doi.org/10.1109/TCAD.2023.3244885}
  {\path{doi:10.1109/TCAD.2023.3244885}}.

\bibitem{Ros21}
Gregory Rosenthal.
\newblock Query and depth upper bounds for quantum unitaries via {Grover}
  search.
\newblock ArXiv e-prints, 2021.
\newblock \href {http://arxiv.org/abs/2111.07992} {\path{arXiv:2111.07992}}.

\bibitem{ZLY22}
Xiao-Ming Zhang, Tongyang Li, and Xiao Yuan.
\newblock Quantum state preparation with optimal circuit depth: implementations
  and applications.
\newblock {\em Physical Review Letters}, 129(23):230504, 2022.
\newblock \href {https://doi.org/10.1103/PhysRevLett.129.230504}
  {\path{doi:10.1103/PhysRevLett.129.230504}}.

\bibitem{YZ23}
Pei Yuan and Shengyu Zhang.
\newblock Optimal (controlled) quantum state preparation and improved unitary
  synthesis by quantum circuits with any number of ancillary qubits.
\newblock {\em Quantum}, 7:956, 2023.
\newblock \href {https://doi.org/10.22331/q-2023-03-20-956}
  {\path{doi:10.22331/q-2023-03-20-956}}.

\bibitem{MN01}
Cristopher Moore and Martin Nilsson.
\newblock Parallel quantum computation and quantum codes.
\newblock {\em SIAM Journal on Computing}, 31(3):799--815, 2001.
\newblock \href {https://doi.org/10.1137/S0097539799355053}
  {\path{doi:10.1137/S0097539799355053}}.

\bibitem{GHP00}
Frederic Green, Steven Homer, and Christopher Pollett.
\newblock On the complexity of quantum {ACC}.
\newblock In {\em Proceedings of the 15th Annual IEEE Conference on
  Computational Complexity}, pages 250--262, 2000.
\newblock \href {https://doi.org/10.1109/CCC.2000.856756}
  {\path{doi:10.1109/CCC.2000.856756}}.

\bibitem{GHMP02}
Frederic Green, Steven Homer, Cristopher Moore, and Christopher Pollett.
\newblock Counting, fanout and the complexity of quantum {ACC}.
\newblock {\em Quantum Information and Computation}, 2(1):35--65, 2002.
\newblock \href {https://doi.org/10.26421/QIC2.1-3}
  {\path{doi:10.26421/QIC2.1-3}}.

\bibitem{TD04}
Barbara~M. Terhal and David~P. DiVincenzo.
\newblock Adaptive quantum computation, constant depth quantum circuits and
  {Arthur-Merlin} games.
\newblock {\em Quantum Information and Computation}, 4(2):134--145, 2004.
\newblock \href {https://doi.org/10.26421/QIC4.2-5}
  {\path{doi:10.26421/QIC4.2-5}}.

\bibitem{FGHZ05}
Stephen Fenner, Frederic Green, Steven Homer, and Yong Zhang.
\newblock Bounds on the power of constant-depth quantum circuits.
\newblock In {\em Proceedings of the 15th International Symposium on
  Fundamentals of Computation Theory}, pages 44--55, 2005.
\newblock \href {https://doi.org/10.1007/11537311_5}
  {\path{doi:10.1007/11537311_5}}.

\bibitem{HS05}
Peter Høyer and Robert {\v{S}}palek.
\newblock Quantum fan-out is powerful.
\newblock {\em Quantum Information and Computation}, 1(5):81--103, 2005.
\newblock \href {https://doi.org/10.4086/toc.2005.v001a005}
  {\path{doi:10.4086/toc.2005.v001a005}}.

\bibitem{TT16}
Yasuhiro Takahashi and Seiichiro Tani.
\newblock Collapse of the hierarchy of constant-depth exact quantum circuits.
\newblock {\em Computational Complexity}, 25:849--881, 2016.
\newblock \href {https://doi.org/10.1007/s00037-016-0140-0}
  {\path{doi:10.1007/s00037-016-0140-0}}.

\bibitem{Joz06}
Richard Jozsa.
\newblock An introduction to measurement based quantum computation.
\newblock In {\em Quantum Information Processing: From Theory to Experiment},
  volume 199 of {\em NATO Science Series, III: Computer and Systems Sciences},
  pages 137--158. IOS Press, 2006.

\bibitem{Aar05}
Scott Aaronson.
\newblock Ten semi-grand challenges for quantum computing theory, 2005.
\newblock URL: \url{https://www.scottaaronson.com/writings/qchallenge.html}.

\bibitem{CCL23}
Nai-Hui Chia, Kai-Min Chung, and Ching-Yi Lai.
\newblock On the need for large quantum depth.
\newblock {\em Journal of the ACM}, 70(1):6:1--38, 2023.
\newblock \href {https://doi.org/10.1145/3570637} {\path{doi:10.1145/3570637}}.

\bibitem{CM20}
Matthew Coudron and Sanketh Menda.
\newblock Computations with greater quantum depth are strictly more powerful
  (relative to an oracle).
\newblock In {\em Proceedings of the 52nd Annual ACM SIGACT Symposium on Theory
  of Computing}, pages 889--901, 2020.
\newblock \href {https://doi.org/10.1145/3357713.3384269}
  {\path{doi:10.1145/3357713.3384269}}.

\bibitem{AGS22}
Atul~Singh Arora, Alexandru Gheorghiu, and Uttam Singh.
\newblock Oracle separations of hybrid quantum-classical circuits.
\newblock ArXiv e-prints, 2022.
\newblock \href {http://arxiv.org/abs/2201.01904} {\path{arXiv:2201.01904}}.

\bibitem{HLG22}
Atsuya Hasegawa and Fran\c{c}ois Le~Gall.
\newblock An optimal oracle separation of classical and quantum hybrid schemes.
\newblock In {\em Proceedings of the 33rd International Symposium on Algorithms
  and Computation}, pages 6:1--6:14, 2022.
\newblock \href {https://doi.org/10.4230/LIPIcs.ISAAC.2022.6}
  {\path{doi:10.4230/LIPIcs.ISAAC.2022.6}}.

\bibitem{ACC+23}
Atul~Singh Arora, Andrea Coladangelo, Matthew Coudron, Alexandru Gheorghiu,
  Uttam Singh, and Hendrik Waldner.
\newblock Quantum depth in the random oracle model.
\newblock In {\em Proceedings of the 55th Annual ACM Symposium on Theory of
  Computing}, pages 1111--1124, 2023.
\newblock \href {https://doi.org/10.1145/3564246.3585153}
  {\path{doi:10.1145/3564246.3585153}}.

\bibitem{BK09}
Anne Broadbent and Elham Kashefi.
\newblock Parallelizing quantum circuits.
\newblock {\em Theoretical Computer Science}, 410(26):2489--2510, 2009.
\newblock \href {https://doi.org/10.1016/j.tcs.2008.12.046}
  {\path{doi:10.1016/j.tcs.2008.12.046}}.

\bibitem{BKP10}
Dan Browne, Elham Kashefi, and Simon Perdrix.
\newblock Computational depth complexity of measurement-based quantum
  computation.
\newblock In {\em Proceedings of the 5th Conference on Quantum Computation,
  Communication, and Cryptography}, pages 35--46, 2010.
\newblock \href {https://doi.org/10.1007/978-3-642-18073-6_4}
  {\path{doi:10.1007/978-3-642-18073-6_4}}.

\bibitem{YF09}
Mingsheng {Y}ing and Yuan {F}eng.
\newblock An algebraic language for distributed quantum computing.
\newblock {\em IEEE Transactions on Computers}, 58(6):728--743, 2009.
\newblock \href {https://doi.org/10.1109/TC.2009.13}
  {\path{doi:10.1109/TC.2009.13}}.

\bibitem{BBG+13}
Robert {B}eals, Stephen {B}rierley, Oliver {G}ray, Aram~W. {H}arrow, Samuel
  {K}utin, Noah {L}inden, Dan {S}hepherd, and Mark {S}tather.
\newblock Efficient distributed quantum computing.
\newblock {\em Proceedings of the Royal Society A: Mathematical, Physical and
  Engineering Sciences}, 469(2153):20120686, 2013.
\newblock \href {https://doi.org/10.1098/rspa.2012.0686}
  {\path{doi:10.1098/rspa.2012.0686}}.

\bibitem{YZLF22}
Mingsheng Ying, Li~Zhou, Yangjia Li, and Yuan Feng.
\newblock A proof system for disjoint parallel quantum programs.
\newblock {\em Theoretical Computer Science}, 897:164--184, 2022.
\newblock \href {https://doi.org/10.1016/j.tcs.2021.10.025}
  {\path{doi:10.1016/j.tcs.2021.10.025}}.

\bibitem{Pre18}
John Preskill.
\newblock Quantum computing in the {NISQ} era and beyond.
\newblock {\em Quantum}, 2:79, 2018.
\newblock \href {https://doi.org/10.22331/q-2018-08-06-79}
  {\path{doi:10.22331/q-2018-08-06-79}}.

\bibitem{JLP+23}
David Jennings, Matteo Lostaglio, Sam Pallister, Andrew~T. Sornborger, and
  Yiğit Subaşı.
\newblock Efficient quantum linear solver algorithm with detailed running
  costs.
\newblock ArXiv e-prints, 2023.
\newblock \href {http://arxiv.org/abs/2305.11352} {\path{arXiv:2305.11352}}.

\bibitem{Dal24}
Alexander~M. Dalzell.
\newblock A shortcut to an optimal quantum linear system solver.
\newblock ArXiv e-prints, 2024.
\newblock \href {http://arxiv.org/abs/2406.12086} {\path{arXiv:2406.12086}}.

\bibitem{CCD+03}
Andrew~M. Childs, Richard Cleve, Enrico Deotto, Edward Farhi, Sam Gutmann, and
  Daniel~A. Spielman.
\newblock Exponential algorithmic speedup by a quantum walk.
\newblock In {\em Proceedings of the 35th Annual ACM Symposium on Theory of
  Computing}, pages 59--68, 2003.
\newblock \href {https://doi.org/10.1145/780542.780552}
  {\path{doi:10.1145/780542.780552}}.

\end{thebibliography}

    \appendix

    \section{Perturbations of Linear Systems}

    \begin{lemma}
    \label{lemma:perturb-QLSP}
        Suppose that $A \in \mathbb{C}^{N \times N}$ is an Hermitian matrix such that $I/\kappa \leq A \leq I$ for some $\kappa \geq 1$. 
        Let $B \in \mathbb{C}^{N \times N}$ be an Hermitian matrix such that $B \leq I$ and $\Abs{A - B} \leq \varepsilon$, where $0 < \varepsilon < 1/\kappa$. 
        Then, 
        \[
        \Abs{A^{-1} - B^{-1}} \leq \frac{\kappa^2\varepsilon}{1-\kappa\varepsilon}. 
        \]
        Moreover, for every unit vector $\ket{x} \in \mathbb{C}^{N}$, i.e., $\Abs{\ket{x}} = 1$, 
        \[
        \Abs*{ \frac{A^{-1}\ket{x}}{\Abs{A^{-1}\ket{x}}} - \frac{B^{-1}\ket{x}}{\Abs{B^{-1}\ket{x}}} } \leq \frac{\kappa^2 \rbra{\kappa + 1}\varepsilon}{1 - \kappa\varepsilon}.
        \]
    \end{lemma}
    \begin{proof}
        Since $I/\kappa \leq A \leq I$, we have $I \leq A^{-1} \leq \kappa I$ and thus $\Abs{A^{-1}} \leq \kappa$.
        Also, since $\Abs{A-B} \leq \varepsilon$, we have $\rbra{1/\kappa - \varepsilon} I \leq B$ and thus $B^{-1} \leq \kappa I/\rbra{1-\kappa\varepsilon}$, therefore $\Abs{B^{-1}} \leq \kappa/\rbra{1-\kappa\varepsilon}$.
        \begin{align*}
            \Abs*{A^{-1} - B^{-1}}
            & = \Abs*{ A^{-1} \rbra*{ I - AB^{-1} } } \\
            & \leq \Abs*{A^{-1}} \Abs*{ I - AB^{-1} } \\
            & = \Abs*{A^{-1}} \Abs*{ BB^{-1} - AB^{-1} } \\
            & \leq \Abs*{A^{-1}} \Abs*{B - A} \Abs*{B^{-1}} \\
            & \leq \kappa \cdot \varepsilon \cdot \frac{\kappa}{1-\kappa\varepsilon} \\
            & = \frac{\kappa^2\varepsilon}{1-\kappa\varepsilon}.
        \end{align*}

        For every unit vector $\ket{x}$, we have the following properties:
        \[
        1 \leq \Abs*{A^{-1} \ket{x}} \leq \kappa,
        \]
        \[
        1 \leq \Abs*{B^{-1} \ket{x}} \leq \frac{\kappa}{1-\kappa\varepsilon},
        \]
        \[
        \Abs*{ A^{-1} \ket{x} - B^{-1} \ket{x} } \leq \Abs*{ A^{-1} - B^{-1} } \leq \frac{\kappa^2\varepsilon}{1-\kappa\varepsilon}.
        \]
        Therefore, we further have
        \begin{align*}
            & \Abs*{ \frac{A^{-1}\ket{x}}{\Abs{A^{-1}\ket{x}}} - \frac{B^{-1}\ket{x}}{\Abs{B^{-1}\ket{x}}} } \\
            & \leq \Abs*{ \frac{A^{-1}\ket{x}}{\Abs{A^{-1}\ket{x}}} - \frac{A^{-1}\ket{x}}{\Abs{B^{-1}\ket{x}}} } \\
            & \qquad + \Abs*{ \frac{A^{-1}\ket{x}}{\Abs{B^{-1}\ket{x}}} - \frac{B^{-1}\ket{x}}{\Abs{B^{-1}\ket{x}}} } \\
            & \leq \abs*{ \frac{1}{\Abs{A^{-1}\ket{x}}} - \frac{1}{\Abs{B^{-1}\ket{x}}} } \Abs*{A^{-1}\ket{x}} \\
            & \qquad + \frac{1}{\Abs{B^{-1}\ket{x}}} \Abs*{ A^{-1}\ket{x} - B^{-1}\ket{x} } \\
            & \leq \Abs*{ A^{-1}\ket{x} - B^{-1}\ket{x} } \rbra*{ \Abs*{A^{-1}\ket{x}} + \frac 1 {\Abs{B^{-1}\ket{x}}} } \\
            & \leq \frac{\kappa^2\varepsilon}{1-\kappa\varepsilon} \cdot \rbra*{ \kappa + 1 } \\
            & \leq \frac{\kappa^2 \rbra{\kappa + 1}\varepsilon}{1 - \kappa\varepsilon}.
        \end{align*}
    \end{proof}

\end{document}